\documentclass[conference]{IEEEtran}

\usepackage{amsmath,amssymb,amsthm,mathtools}
\usepackage{algorithm,algpseudocode}
\usepackage{graphicx}
\usepackage{hyperref}
\usepackage{booktabs}
\usepackage{microtype}
\usepackage{xcolor}
\usepackage{enumitem}
\usepackage[numbers,sort&compress]{natbib}
\usepackage{balance}

\hypersetup{
	colorlinks=true,
	linkcolor=blue!60!black,
	citecolor=blue!60!black,
	urlcolor=blue!60!black
}

\setlist{nosep,leftmargin=*}
\makeatletter
\def\thm@space@setup{%
	\thm@preskip=4pt plus 1pt minus 1pt
	\thm@postskip=\thm@preskip
}
\let\oldthebibliography\thebibliography
\renewcommand{\thebibliography}[1]{%
	\oldthebibliography{#1}%
	\setlength{\itemsep}{0pt}%
	\setlength{\parsep}{0pt}%
	\setlength{\parskip}{0pt}%
}
\makeatother

\newtheorem{theorem}{Theorem}
\newtheorem{lemma}[theorem]{Lemma}

\theoremstyle{definition}
\newtheorem{definition}[theorem]{Definition}

\newtheorem{remark}[theorem]{Remark}

\newcommand{\KL}{D_{\mathrm{KL}}}
\newcommand{\KLL}{D_{\mathrm{KL},L}}

\newcommand{\CMI}{\mathrm{I}}
\newcommand{\calX}{\mathcal{X}}
\newcommand{\calY}{\mathcal{Y}}
\newcommand{\calZ}{\mathcal{Z}}

\newcommand{\bbE}{\mathbb{E}}
\newcommand{\bbN}{\mathbb{N}}
\newcommand{\bbP}{\mathbb{P}}

\newcommand{\supp}{\mathrm{supp}}

\newcommand{\Order}{\mathcal{O}}

\begin{document}
	
	\title{Quantum Causal Discovery via Amplitude Estimation of Kullback--Leibler Divergence}
	
	\author{
		\IEEEauthorblockN{Shabnam Sodagari, \textit{Senior Member, IEEE}}
		\IEEEauthorblockA{California State University,
			Long Beach, CA 90840\\
			Email: shabnam@csulb.edu}
	}
	
	\IEEEaftertitletext{\vspace{-0.75\baselineskip}}
	
	\maketitle
	
	\begin{abstract}
		Causal discovery from observational data underpins applications in
		finance, climate modeling, and machine learning. Constraint-based
		causal discovery reduces structure learning to a sequence of
		conditional independence (CI) tests, where each test decides
		independence by estimating conditional mutual information
		$I(X;Y \mid Z)$ to additive precision $\tau$ and thresholding against
		it. Classically this requires $\Theta(1/\tau^{2})$ samples per test,
		a cost that dominates in the high-precision regime typical of weak
		dependencies. We present QKLA (Quantum Kullback--Leibler Amplitude
		estimation), a quantum algorithm that encodes a clipped log-density
		ratio as a bounded amplitude and applies amplitude estimation to
		recover a clipped KL expectation. Given coherent oracle access to the
		relevant distributions and a reversible log-ratio arithmetic oracle,
		QKLA achieves a quadratic precision improvement, needing only
		$\mathcal{O}((L/\tau)\log(1/\delta))$ queries, where $L$ is the
		log-ratio clip bound. Under per-stratum conditional-oracle access and
		a margin assumption for CI decisions, embedding this estimator in the
		PC algorithm compounds to an $\widetilde{\Omega}(1/(L\tau))$
		reduction in total oracle queries. We validate the theory in three
		experiments. A gate-level state-vector simulation of the full QKLA
		circuit confirms the predicted $\mathcal{O}(1/M)$ error decay.
		Across $K=20$ random binary distributions, classical and quantum
		error scalings match theory to within $0.01$ in slope. In an
		oracle-model benchmark inside PC on two networks (\textsc{Asia}, 8
		nodes; \textsc{Synthetic-12}, 12 nodes), the quantum CI subroutine
		reaches comparable skeleton-recovery $F_1$ while using
		$2.7$--$3.2\times$ fewer oracle queries at
		$\tau = 5\cdot 10^{-3}$ bits and $4.0$--$7.4\times$ fewer at
		$\tau = 10^{-3}$ bits.
	\end{abstract}
	
	\section{Introduction}
	
	Constraint-based causal discovery recovers the skeleton of a causal graph
	from a sequence of conditional independence (CI) tests on a joint
	distribution~\citep{spirtes2000causation,pearl2009causality}.
	The PC algorithm~\citep{spirtes1991algorithm} issues at most
	$\Order(n^{d+2})$ tests for $n$ variables at conditioning depth
	$d$~\citep{spirtes2000causation}, and each is nontrivial:
	plug-in estimation of $\CMI(X;Y\mid Z)$ to additive precision $\tau$
	requires $\Theta(1/\tau^2)$ i.i.d.\ samples in the fixed-support
	regime~\citep{paninski2003estimation}, while CI testing
	(binary $X,Y$) has matching complexity
	$\Theta\!\bigl(\max\bigl(\sqrt{|\calZ|}/\tau^2,\ \min(|\calZ|^{7/8}/\tau,\,
	|\calZ|^{6/7}/\tau^{8/7})\bigr)\bigr)$~\citep{canonne2018testing},
	where $|\calZ|$ is the number of distinct values of the conditioning
	vector $Z$. In the high-precision regime $\tau \in [10^{-3},10^{-2}]$
	bits, this $1/\tau^2$ classical cost dominates.
	
	Quantum amplitude estimation (QAE)~\citep{brassard2002quantum} reduces
	the query complexity of expectation estimation from $\Order(1/\tau^2)$
	samples to $\Order(1/\tau)$ oracle queries, a quadratic improvement in
	precision. Montanaro~\citep{montanaro2015quantum} generalized this into
	a framework for quantum speedup of Monte Carlo methods. Since conditional independence (CI) testing in its estimation
	formulation reduces to computing an expectation---the KL divergence
	$\KL(p(X,Y\mid z)\,\|\,p(X\mid z)\,p(Y\mid z)) = \bbE_p[\log(p/q)]$
	(with $p, q$ denoting the joint and product conditionals
	respectively)---the natural question is whether QAE transfers this
	$1/\tau^{2} \to 1/\tau$ advantage to constraint-based causal
	discovery. We answer affirmatively: in the oracle query model, we derive
	rigorous per-test and compound complexity bounds and make the
	governing constants---particularly the log-ratio clip $L$---explicit. 
	
	The contributions of this work are:
	\begin{enumerate}[leftmargin=*]
		\item We give a quantum algorithm \textsc{QKLA} that estimates the clipped KL
		expectation $\KLL(p \| q)$ to additive precision $\tau$ with success
		probability $1 - \delta$ using
		$M \cdot k$ calls to a preparation oracle for $p$ and a reversible
		log-ratio arithmetic oracle, where
		$M = 2^{\lceil \log_2 \lceil 4\pi L/\tau \rceil \rceil}$ Grover iterates
		satisfy $M \le 2\lceil 4\pi L/\tau\rceil$ and
		$k = \lceil 5\ln(1/\delta) \rceil$ median-amplification rounds, giving
		$M \cdot k = \Order\!\bigl((L/\tau)\log(1/\delta)\bigr)$ total queries
		(Theorem~\ref{thm:qkla}). The algorithm encodes the clipped log-ratio
		as a bounded amplitude via a uniformly controlled rotation and then
		applies QAE to the resulting $|1\rangle$-probability
		(Section~\ref{sec:algorithm}). 
		
		\item Assuming per-stratum conditional preparation oracles and
		classical access to the stratum weights $p(z)$, we lift
		\textsc{QKLA} to a conditional mutual information estimator
		\textsc{QCMIE} at query cost
		$\Order(|\calZ|\,L/\tau\,\log(|\calZ|/\delta))$
		(Theorem~\ref{thm:qcmi}), and embed it in PC to obtain a compound
		bound of $\widetilde{O}(n^{d+2}\, r^d\, L/\tau)$ quantum queries
		(Theorem~\ref{thm:compound}) versus
		$\widetilde\Omega(n^{d+2}\,r^d /\tau^2)$ classical samples in the
		fresh-data-per-test regime
		(Sections~\ref{sec:cmi}--\ref{sec:compound}).
		
		\item We perform gate-level validation by simulating the full
		\textsc{QKLA} circuit on a state-vector simulator: explicit
		construction of the preparation oracle via Gram--Schmidt,
		reversible log-ratio arithmetic with $n_{\mathrm{arith}} = 6$
		qubits, controlled Grover iterates, and inverse QFT. The simulated
		phase-register distribution matches the canonical theoretical QAE
		distribution to machine precision, and the $\Order(1/M)$ decay of
		the $80$th-percentile error is confirmed at log-log slope $-1.38$
		(Section~\ref{sec:exp-statevec}).
		
		\item In precision-scaling and PC benchmarks, averaged over
		$K = 20$ random binary distributions, we measure classical slope
		$-0.500$ and quantum slope $-1.009$, matching the theoretical
		$N^{-1/2}$ and $M^{-1}$ rates to within $0.01$ in slope
		(Fig.~\ref{fig:precision}); the crossover in queries to target
		precision sits near $\tau \approx 1.7 \cdot 10^{-2}$ bits. In an
		oracle-model benchmark for the CI subroutine inside PC on
		\textsc{Asia} ($8$ binary nodes, $8$ directed
		edges)~\cite{lauritzen1988local} and \textsc{Synthetic-12}---a
		randomly generated $12$-node binary DAG with edge probability $0.22$,
		random seed $11$, and $23$ directed edges---the quantum subroutine
		reaches the same skeleton-recovery $F_1$ as the classical plug-in
		estimator while using $2.7$--$3.2\times$ fewer oracle queries at
		$\tau = 5 \cdot 10^{-3}$ bits and $4.0$--$7.4\times$ fewer at
		$\tau = 10^{-3}$ bits (Fig.~\ref{fig:pc}), with the advantage growing
		as $\tau \to 0$ as predicted by Theorem~\ref{thm:compound}
		(Sections~\ref{sec:exp-precision}--\ref{sec:exp-pc}).
	\end{enumerate}
	
	Prior quantum KL-divergence estimation work~\citep{li2019quantum}
	operates in a different oracle model---the BHH coherent
	sampling-oracle~\citep{bravyi2011quantum}---and estimates
	$\KL(p \| q)$ at rate $\widetilde{\mathcal{O}}(\sqrt{s}/\varepsilon^{2})$
	in alphabet size $s$ and precision $\varepsilon$. Our setting differs on
	three axes. First, we use a preparation oracle plus a reversible
	log-ratio arithmetic oracle $\mathcal{O}_{\log}^{p,q}$, which computes
	the log-ratio coherently on the full superposition rather than
	reconstructing it from per-bin amplitude-estimation calls. Second,
	\textsc{QKLA} makes a \emph{single} QAE invocation rather than nesting
	quantum counting inside an outer Monte Carlo loop, yielding the
	$\mathcal{O}(L/\tau)$ per-call rate of Theorem~\ref{thm:qkla}---quadratically
	tighter in precision than the $1/\varepsilon^{2}$ scaling
	of~\citep{li2019quantum}, at the cost of an explicit factor $L$ in the log-ratio
	clip. Third, the two results target complementary regimes: the
	bound in~\cite{li2019quantum} is tight for large-alphabet distribution
	testing, while our \textsc{QKLA} targets the small-alphabet,
	high-precision regime of constraint-based CI testing, where the alphabet
	is $\mathcal{O}(1)$ and $\tau$ is the dominant parameter. Under the
	per-stratum-oracle assumptions of Section~\ref{sec:cmi}, we
	compose \textsc{QKLA} into a CMI estimator and then into PC
	(Theorems~\ref{thm:qcmi} and~\ref{thm:compound}); neither composition
	appears in prior quantum divergence-estimation work.
	
	Section~\ref{sec:preliminaries} contains notation and recalls canonical
	QAE. Section~\ref{sec:algorithm} presents \textsc{QKLA} and its per-call
	bound. Section~\ref{sec:cmi} composes per-call estimates into
	\textsc{QCMIE}. Section~\ref{sec:compound} proves the PC compound bound.
	Section~\ref{sec:experiments} reports experiments.
	Section~\ref{sec:related} places the result in context, and
	Sections~\ref{sec:discussion} and \ref{sec:conclus} conclude.
	
	\section{Preliminaries}
	\label{sec:preliminaries}
	
	\subsection{Conditional independence via conditional mutual information}
	
	Let $(X, Y, Z)$ be a jointly distributed triple on finite alphabets
	$\calX, \calY, \calZ$ with joint $p(x, y, z)$. We write $X \perp Y \mid Z$
	if $p(x, y \mid z) = p(x \mid z)\,p(y \mid z)$ for every $(x, y, z)$ with
	$p(z) > 0$. The conditional mutual information (CMI) is
	\begin{equation}
		\footnotesize
		\CMI(X; Y \mid Z) \;=\;
		\sum_{z \in \calZ} p(z)\,
		\KL\!\bigl(\,p(X, Y \mid z)\,\|\,p(X \mid z)\,p(Y \mid z)\,\bigr).
		\label{eq:cmi-as-kl}
	\end{equation}
	$\CMI \geq 0$ with equality iff $X \perp Y \mid Z$~\citep{cover2006elements}.
	We write $r$ for the maximum alphabet size per variable and, when $Z$ is
	a $d$-dimensional vector, $|\calZ| \leq r^d$. Throughout, $\log$ denotes
	$\log_2$ unless otherwise stated.
	
	\subsection{The oracle model}
	\label{subsec:oracles}
	
	We work in the standard quantum query model. Let $n_q$ denote the number
	of qubits encoding a sample from $\calX \times \calY \times \calZ$.
	
	\begin{definition}[Preparation oracle $\Order_p$]
		\label{def:prep}
		There exists a unitary $\Order_p$ on $n_q + \Order(1)$ qubits\footnote{The
			$\Order(1)$ ancilla qubits are scratch space used internally by the
			oracle's circuit and returned to $|0\rangle$ at the end of the
			computation; they do not appear in the output state.} with
		\begin{equation}
			\Order_p\,|0\rangle \;=\; |\psi_p\rangle \;=\;
			\sum_{(x, y, z) \in \calX \times \calY \times \calZ}
			\sqrt{p(x, y, z)}\,|x, y, z\rangle.
			\label{eq:prep}
		\end{equation}
		This is the standard access model in quantum
		distribution testing and Monte Carlo~\citep{grover2002creating,
			bravyi2011quantum, montanaro2015quantum}.
	\end{definition}
	
	\begin{definition}[Log-ratio arithmetic oracle $\Order_{\log}$]
		\label{def:logoracle}
		Fix a clip bound $L > 0$ and a fixed-point precision
		$b = \Order(\log(L/\tau))$ bits. There exists a reversible unitary
		$\Order_{\log}^{p,q}$ acting on the $n_q$-qubit sample register and a
		$b$-qubit output register, such that for every basis state
		$|w\rangle$ with $w \in \calX \times \calY \times \calZ$,
		\begin{align}
			&\Order_{\log}^{p,q}\,|w\rangle\,|0\rangle_b \;=\;
			|w\rangle\,\bigl|\mathrm{clip}_L\!\bigl(\log_2(p(w)/q(w))\bigr)\bigr\rangle_b,
			\\
			&\mathrm{clip}_L(u) := \max(-L, \min(L, u)) \text{ for } u \in \mathbb{R}.
			\label{eq:logoracle}
		\end{align}
		The sample register is left unchanged and the clipped log-ratio is
		written, to $b$ bits of fixed-point precision, into the previously-zero
		output register. Writing
		$\ell_L(w) := \mathrm{clip}_L\!\bigl(\log_2(p(w)/q(w))\bigr)$
		for the clipped log-ratio, linearity of $\Order_{\log}^{p,q}$ gives
		\begin{equation*}
			\Order_{\log}^{p,q}\Bigl(\sum_w \alpha_w |w\rangle\Bigr)|0\rangle_b
			\;=\; \sum_w \alpha_w |w\rangle\,|\ell_L(w)\rangle_b,
		\end{equation*}
		so the log-ratio is evaluated coherently across all samples in a single
		query.
	\end{definition}
	
	We interpret the unclipped log-ratio with the conventions
	$\log_2(p(w)/0)=+\infty$ when $p(w)>0$,
	$\log_2(0/q(w))=-\infty$ when $q(w)>0$,
	and assign an arbitrary fixed value (e.g.\ $0$) when $p(w)=q(w)=0$,
	since such basis states never occur under the preparation oracle for $p$.
	After clipping, these cases map to $+L$, $-L$, and the chosen dummy value,
	respectively.
	
	The map $|w\rangle|0\rangle \mapsto |w\rangle|f(w)\rangle$ is reversible
	for any function $f$ because $|w\rangle$ is preserved. Reversibility is
	essential because Algorithm~\ref{alg:qkla} applies
	$(\Order_{\log}^{p,q})^\dagger$ in step~6 to uncompute the log-ratio
	register; without uncomputation, the arithmetic register would remain
	entangled with the sample register and the amplitude encoding of step~5
	would not produce the target state.
	
	The output register discretizes $[-L,+L]$ into $2^b$ levels with
	resolution $\Order(2L/2^b)$. If the clipped log-ratio $\ell_L$ is
	represented to additive error at most
	$\Delta_b = \Order(2L/2^b)$ uniformly over basis states, then the
	induced error in $\KLL = \bbE_p[\ell_L]$ is also at most $\Delta_b$.
	Therefore it suffices to choose $b$ so that $\Delta_b \lesssim \tau$,
	i.e.\ $b = \Order(\log(L/\tau))$.
	
	The clip ensures the output value lies in a bounded interval
	$[-L,+L]$, so the subsequent affine map
	$g_L(w) = (\ell_L(w)+L)/(2L)\in[0,1]$ (Eq.~\ref{eq:g-def}) yields
	a valid amplitude. Without the clip,
	$\log_2(p(w)/q(w))$ could be arbitrarily large in magnitude whenever
	$p$ or $q$ assigns small mass to some $w$, and no bounded amplitude
	encoding would exist. The resulting difference relative to the true KL
	is controlled by Lemma~\ref{lem:bias} when $\supp(p)\subseteq\supp(q)$.
	
	Definitions~\ref{def:prep}--\ref{def:logoracle} are standard
	access models in quantum distribution testing and Monte
	Carlo~\citep{bravyi2011quantum, montanaro2015quantum,
		gilyen2020distributional}.
	We discuss gate-level implementations in Section~\ref{sec:discussion}.
	Our bounds count oracle queries; under standard sparse-access or QRAM input
	models, gate-level implementations incur polylogarithmic overhead per
	query~\citep{berry2015hamiltonian}.
	
	We use canonical phase-estimation-based QAE~\citep{brassard2002quantum};
	iterative variants~\citep{grinko2021iterative, suzuki2020amplitude} yield
	identical asymptotics with different constants, and our results carry
	over by substitution.
	
	\begin{theorem}[Amplitude estimation, {\citep[Theorem~12]{brassard2002quantum}}]
		\label{thm:qae}
		Let $\mathcal{A}$ be a unitary on $n_q + 1$ qubits with
		\begin{equation*}
			\mathcal{A}\,|0\rangle|0\rangle \;=\;
			\sqrt{1-a}\,|\phi_0\rangle|0\rangle
			\;+\; \sqrt{a}\,|\phi_1\rangle|1\rangle
		\end{equation*}
		for some $a \in [0, 1]$. For any $M = 2^t \in \bbN$, there is a
		quantum algorithm (\textsc{EstAmp}) that makes $M$ applications of
		the Grover operator
		$\mathcal{G} = -\mathcal{A}\,S_0\,\mathcal{A}^{\dagger}\,S_{\chi}$
		(where $S_0, S_{\chi}$ are reflections), followed by an inverse
		quantum Fourier transform on $t$ qubits, and outputs
		$\hat a = \sin^{2}(\pi m / M)$ for a measured integer
		$m \in \{0, \ldots, M - 1\}$ satisfying
		\begin{equation}
			|\hat a - a|
			\;\leq\; \frac{2\pi \sqrt{a(1-a)}}{M} \;+\; \frac{\pi^{2}}{M^{2}}
			\label{eq:qae-error}
		\end{equation}
		with probability at least $8/\pi^{2} \approx 0.811$. If $a = 0$
		then $\hat a = 0$ with certainty, and if $a = 1$ and $M$ is even
		then $\hat a = 1$ with certainty.
	\end{theorem}
	
	\begin{remark}[Median amplification]
		\label{rem:median-amp}
		Running $k$ independent copies of \textsc{EstAmp} and taking the
		median of the outputs drives the failure probability of the bound
		in~\eqref{eq:qae-error} to at most
		$\exp\bigl(-k \cdot D(\tfrac{1}{2} \| 8/\pi^{2})\bigr)
		\leq \exp(-0.24\,k)$
		by a standard Chernoff
		argument~\citep[Ch.~4]{mitzenmacher2017probability}, where
		$D(\cdot\|\cdot)$ denotes binary KL divergence in nats. In
		particular, $k = \lceil 5 \ln(1/\delta) \rceil$ independent runs
		suffice for failure probability at most $\delta$.
	\end{remark}
	
	Equation~\eqref{eq:qae-error} is the exact Brassard--H\o yer--Mosca--Tapp
	bound; the cleaner weakening
	$|\hat a - a| \leq (\pi/M)(2\sqrt{a(1-a)} + \pi/M)$ costs only
	multiplicative constants. The factor $\sqrt{a(1-a)}$ controls the dominant
	error term: it is largest at $a = 1/2$ (where it equals $1/2$) and
	vanishes at the endpoints $a \in \{0,1\}$. This motivates our encoding
	choice in Section~\ref{sec:algorithm}. Exact conditional independence
	implies $\KLL = 0$, which by Equation~\eqref{eq:klL} maps to $a = 1/2$,
	precisely the worst case of the QAE error bound. We therefore encode the
	full interval $a \in [0,1]$, so that the independence case is included
	without any special-case treatment.
	
	\section{The \textsc{QKLA} Algorithm}
	\label{sec:algorithm}
	
	\subsection{Encoding a KL divergence as a bounded amplitude}
	\label{subsec:encoding}
	
	Fix a clip bound $L > 0$. For distributions $p, q$ on $\calX$ define
	\begin{align}
		&\ell_L(x) \;=\; \mathrm{clip}_L\!\bigl(\log_2(p(x)/q(x))\bigr), \nonumber \\
		&g_L(x) \;=\; \frac{\ell_L(x) + L}{2L} \in [0, 1]. \label{eq:g-def}
	\end{align}
	
	The clipped KL expectation is
	\begin{equation}
		\KLL(p \| q) \;:=\; \sum_{x \in \calX} p(x)\,\ell_L(x)
		\;=\; 2L\,\bbE_p[g_L] - L.
		\label{eq:klL}
	\end{equation}
	The clipping difference $\eta_L(p, q) := \KL(p \| q) - \KLL(p \| q)$ is
	controlled by the following lemma when $\supp(p)\subseteq\supp(q)$.
	
	\begin{lemma}[Clipping bias]
		\label{lem:bias}
		Assume $\supp(p)\subseteq\supp(q)$, and let
		$\rho(x) := \log_2(p(x)/q(x))$ for $x \in \supp(p)$. Then the clipping
		difference $\eta_L(p, q) := \KL(p \| q) - \KLL(p \| q)$ satisfies
		\begin{equation}
			|\eta_L(p, q)| \;\leq\;
			\bbE_p\!\bigl[(|\rho(X)| - L)\,\mathbf{1}\{|\rho(X)| > L\}\bigr],
			\label{eq:bias-twosided}
		\end{equation}
		and hence
		\begin{equation}
			|\eta_L(p, q)| \;\leq\;
			\Bigl(\log_2(1/p_{\min}) + \log_2(1/q_{\min})\Bigr)\,
			\Pr_{X \sim p}\!\bigl[|\rho(X)| > L\bigr],
		\end{equation}
		where
		$p_{\min} = \min_{x \in \supp(p)} p(x)$ and
		$q_{\min} = \min_{x \in \supp(q)} q(x)$.
		In particular, if
		$2^{-L}\,p(x) \leq q(x) \leq 2^{L}\,p(x)$ for every $x \in \supp(p)$
		(equivalently, $|\rho(x)| \leq L$ on $\supp(p)$), then
		$\eta_L(p, q) = 0$ exactly and $\KLL(p\|q) = \KL(p\|q)$.
	\end{lemma}
	
	\begin{proof}
		The two-sided sufficient condition is immediate: if $|\rho(x)| \leq L$ on
		$\supp(p)$, then the clip is inactive and $\ell_L(x) = \rho(x)$ pointwise,
		giving $\KLL = \KL$.
		
		For the general case, write
		\[
		\rho - \mathrm{clip}_L(\rho) \;=\;
		\begin{cases}
			\rho - L & \text{if } \rho > L, \\
			0        & \text{if } -L \leq \rho \leq L, \\
			\rho + L & \text{if } \rho < -L.
		\end{cases}
		\]
		Taking expectation under $p$ and using linearity,
		\begin{align}
			\eta_L(p,q)
			&= \bbE_p\!\bigl[\rho - \mathrm{clip}_L(\rho)\bigr] \notag \\
			&= \underbrace{\bbE_p\!\bigl[(\rho - L)\,\mathbf{1}\{\rho > L\}\bigr]}_{=:\, \eta_L^{+} \,\geq\, 0}
			\;+\; \underbrace{\bbE_p\!\bigl[(\rho + L)\,\mathbf{1}\{\rho < -L\}\bigr]}_{=:\, \eta_L^{-} \,\leq\, 0}.
			\label{eq:bias-decomp}
		\end{align}
		The two contributions have opposite signs, so
		\[
		|\eta_L(p, q)| \;\leq\; \eta_L^{+} - \eta_L^{-}
		\;=\; \bbE_p\!\bigl[(|\rho| - L)\,\mathbf{1}\{|\rho| > L\}\bigr],
		\]
		which is~\eqref{eq:bias-twosided}. For the second bound, for every
		$x \in \supp(p)\subseteq\supp(q)$,
		\begin{align}
			|\rho(x)|
			&= \bigl|\log_2 p(x) - \log_2 q(x)\bigr| \notag\\
			&\le |\log_2 p(x)| + |\log_2 q(x)| \notag\\
			&\le \log_2(1/p_{\min}) + \log_2(1/q_{\min}).
		\end{align}
		Hence $|\rho(x)| - L \le |\rho(x)|$ on the event $\{|\rho|>L\}$, and
		pulling the constant outside the expectation gives the claim.
		
		If $\supp(p)\not\subseteq\supp(q)$, then $\KL(p\|q)=+\infty$, so a
		finite-error comparison to $\KL$ is not meaningful; in that case
		\textsc{QKLA} estimates $\KLL$ as defined.
	\end{proof}
	
	\begin{remark}
		$\KLL(p\|q)$ is not a formal divergence: it can be negative, since clipping can cap rare large
		positive log-ratios while leaving moderate negative ones. Also, it is not the image of $\KL$ under any data-processing channel
		on $(p, q)$. It is a clipped log-ratio expectation chosen for its bounded
		amplitude encoding (Section~\ref{subsec:encoding}); the relationship to
		$\KL$ is the bound of Lemma~\ref{lem:bias}.
	\end{remark}
	
	The lemma justifies treating $L$ as a design parameter: choose $L$ at
	least as large as $\max_x |\log_2(p(x)/q(x))|$ to eliminate the
	difference. For the benchmarks of Section~\ref{sec:experiments},
	$L = 3$ exactly covers the ASIA log-ratios; $L = 6$ would be needed to
	exactly cover all $K = 20$ random distributions of Experiment~2, and we
	use $L = 3$ throughout, accepting a small clipping difference on the most
	extreme of those 20 distributions (Lemma~\ref{lem:bias}). When $L$
	exceeds the log-ratio range, $\eta_L = 0$ exactly and $\KLL = \KL$,
	so CI testing via the clipped expectation coincides with CI testing
	via the true KL.
	
	For CI testing via $\KLL$, the relevant implication is: if $\KL = 0$ then
	$\KLL = 0$ exactly (since $\rho \equiv 0$ on $\supp(p)$ and the clip is
	inactive), and conversely if $|\KLL| \leq \tau$ then
	$\KL \leq \tau + |\eta_L|$ by Lemma~\ref{lem:bias}. So the clipped
	expectation remains a sound CI statistic: independence is detected exactly,
	and small $|\KLL|$ implies small $\KL$ up to the controlled clipping
	difference.
	
	\subsection{The algorithm and its analysis}
	
	\begin{algorithm*}[!t]
		\caption{\textsc{QKLA}: Quantum KL Divergence Amplitude Estimation}
		\label{alg:qkla}
		\begin{algorithmic}[1]
			\Require Oracle $\Order_p$ (Def.~\ref{def:prep}); oracle
			$\Order_{\log}^{p,q}$ (Def.~\ref{def:logoracle}); precision $\tau > 0$;
			confidence $\delta \in (0, 1)$; clip bound $L$.
			\Ensure $\widehat\KL$ with $|\widehat\KL - \KLL(p \| q)| \leq \tau$ with
			probability $\geq 1 - \delta$.
			\State $M \gets 2^{\lceil \log_2 \lceil 4\pi L / \tau \rceil \rceil}$;\quad
			$k \gets \lceil 5 \ln(1/\delta) \rceil$.
			\For{$j = 1, \ldots, k$}
			\State Prepare $|\psi_p\rangle \otimes |0\rangle_b \otimes |0\rangle$
			via $\Order_p$, where $|0\rangle_b$ is a $b$-bit arithmetic
			register and the final qubit is the amplitude ancilla.
			\State Apply $\Order_{\log}^{p,q}$: compute $\ell_L(x)$ into the
			arithmetic register.
			\State Apply the uniformly controlled rotation
			$R_y(2\arcsin\sqrt{g_L(\cdot)})$ on the amplitude ancilla,
			controlled by the arithmetic register.
			\State Apply $(\Order_{\log}^{p,q})^\dagger$ to uncompute the
			arithmetic register.
			\State Let $\mathcal{A}_j$ denote the composed unitary. Run
			QAE (Theorem~\ref{thm:qae}) with $M$ Grover iterations of
			$\mathcal{G}_j = -\mathcal{A}_j\,S_0\,\mathcal{A}_j^\dagger\,S_\chi$;
			measure the $t$-qubit phase register to obtain $m_j$; set
			$\hat a^{(j)} = \sin^2(\pi m_j / M)$.
			\EndFor
			\State $\hat a \gets \mathrm{median}\{\hat a^{(1)}, \ldots, \hat a^{(k)}\}$.
			\State \Return $\widehat\KL \gets 2L\,\hat a - L$.
		\end{algorithmic}
	\end{algorithm*}
	
	\begin{theorem}[Per-call query complexity]
		\label{thm:qkla}
		Under Definitions~\ref{def:prep}--\ref{def:logoracle},
		Algorithm~\ref{alg:qkla} outputs $\widehat\KL$ satisfying
		\begin{equation}
			|\widehat\KL - \KLL(p \| q)| \;\leq\; \tau
			\quad \text{with probability at least } 1 - \delta,
			\label{eq:qkla-bound}
		\end{equation}
		using
\begin{align}
	T_{\mathrm{QKLA}}(\tau,\delta)
	&= M \cdot k \notag\\
	&\le 2\lceil 4\pi L/\tau\rceil \cdot \lceil 5\ln(1/\delta)\rceil \notag\\
	&= \Order\!\bigl((L/\tau)\log(1/\delta)\bigr).
\end{align}
		calls to the composed unitary $\mathcal{A}$ and hence the same order
		of queries to the base oracles $\Order_p$ and $\Order_{\log}^{p,q}$
		up to constant factors. If $\supp(p)\subseteq\supp(q)$, then
		\[
		|\widehat\KL - \KL(p\|q)| \le \tau + |\eta_L(p,q)|.
		\]
		Otherwise $\KL(p\|q)=+\infty$, and $\widehat\KL$ estimates $\KLL$ as
		defined.
	\end{theorem}
	
	\begin{proof}
		After step~6, the state is
		\[
		\sum_x \sqrt{p(x)}\,|x\rangle\,|0\rangle_b\,
		\bigl(\sqrt{1 - g_L(x)}\,|0\rangle + \sqrt{g_L(x)}\,|1\rangle\bigr),
		\]
		because the uncomputation disentangles the arithmetic register when
		$\Order_{\log}$ is reversible. The probability of measuring $|1\rangle$
		on the amplitude ancilla is therefore
		$a := \bbE_p[g_L] \in [0, 1]$, and by~\eqref{eq:klL},
		$\KLL(p \| q) = 2L\,a - L$.
		
		A single QAE run with
		$M = 2^{\lceil \log_2 \lceil 4\pi L / \tau \rceil \rceil}$
		(a power of $2$ as required by Theorem~\ref{thm:qae}, satisfying
		$4\pi L/\tau \leq M \leq 2\lceil 4\pi L/\tau \rceil$) returns
		$\hat a^{(j)}$ satisfying
		\[
		|\hat a^{(j)} - a| \leq
		\frac{2\pi \sqrt{a(1-a)}}{M} + \frac{\pi^2}{M^2}
		\leq \frac{\pi}{M} + \frac{\pi^2}{M^2}
		\leq \frac{2\pi}{M}
		\leq \frac{\tau}{2L}
		\]
		with probability $\geq 8/\pi^2$, using $\sqrt{a(1-a)} \leq 1/2$ for the
		first inequality and $\pi^2/M^2 \leq \pi/M$ (equivalent to $M \geq \pi$,
		satisfied since $M \geq 4\pi L/\tau \geq 4\pi$ for $L \geq 1$, $\tau \leq 1$)
		for the second inequality. Therefore
		\[
		|\widehat\KL^{(j)} - \KLL| = 2L\,|\hat a^{(j)} - a| \leq \tau,
		\]
		on the single-run success event.
		
		Each of the $k = \lceil 5 \ln(1/\delta) \rceil$ runs succeeds
		independently with probability $q_0 = 8/\pi^2 > 1/2$. Let $S_k$ be the
		number of successes. The median of
		$\{\widehat\KL^{(1)}, \ldots, \widehat\KL^{(k)}\}$ falls inside
		$[\KLL - \tau, \KLL + \tau]$ whenever at least $\lceil k/2 \rceil$
		runs succeed, so
		\begin{align*}
			\bbP[\text{failure}]
			\;&\leq\; \bbP[S_k \leq k/2]
			\;\leq\; \exp  \bigl(-k \cdot D(1/2 \,\|\, 8/\pi^2)\bigr) \\
			&\leq\; \exp(-0.243\,k)
		\end{align*}
		by the Chernoff--Cram\'er bound~\cite{mitzenmacher2017probability}, where
		$D(1/2 \,\|\, 8/\pi^2) \approx 0.24$ nats is computed directly.
		Choosing $k = \lceil 5 \ln(1/\delta) \rceil$ gives failure
		probability at most
		\[
		\exp \bigl(-0.24 \cdot 5 \cdot \ln(1/\delta)\bigr)
		\;=\; \delta^{0.24 \cdot 5}
		\;\approx\; \delta^{1.2}
		\;\leq\; \delta
		\]
		as claimed.
		
		The query count is $M$ Grover iterations per run, giving
		$M \cdot k$ calls to $\mathcal{A}$ up to constant factors.
		If $\supp(p)\subseteq\supp(q)$, the total error with respect to
		$\KL(p \| q)$ adds $|\eta_L(p, q)|$ from Lemma~\ref{lem:bias} via the
		triangle inequality. If $\supp(p)\not\subseteq\supp(q)$, then
		$\KL(p\|q)=+\infty$ and only the $\KLL$ guarantee is meaningful.
	\end{proof}
	
	\begin{remark}[Power-of-two rounding and iterative QAE]
		\label{rem:power-of-two}
		The rounding
		$M \gets 2^{\lceil \log_2 \lceil 4\pi L/\tau \rceil \rceil}$
		in Algorithm~\ref{alg:qkla} satisfies the canonical-QAE hypothesis
		$M = 2^t$ of Theorem~\ref{thm:qae} and the lower bound
		$M \geq 4\pi L/\tau$ used in the proof of Theorem~\ref{thm:qkla}; since
		$M \leq 2\lceil 4\pi L/\tau \rceil$, the rounding at most doubles the
		per-run query count and preserves the $\Order(L/\tau)$ scaling.
		Iterative QAE~\citep{grinko2021iterative} removes the power-of-$2$
		requirement entirely at the cost of an $\Order(\log\log(1/\tau))$
		overhead in total queries, yielding $\widetilde{\Order}(L/\tau)$ per
		call; this is the recommended NISQ-friendly instantiation
		(Section~\ref{sec:discussion}).
	\end{remark}
	
	\begin{remark}[The $\sqrt{a(1-a)}$ factor]
		The bound~\eqref{eq:qae-error} tightens near $a = 0$ and $a = 1$.
		At exact independence, $\KLL = 0$ and therefore $a = 1/2$, the worst
		case of $\sqrt{a(1-a)}$. We state Theorem~\ref{thm:qkla} in this
		uniform worst-case form, which gives a simple instance-independent
		guarantee valid over the entire parameter range. For amplitudes away
		from $a=1/2$, the required single-run value of $M$ can be up to a
		factor of two smaller, so the theorem should be viewed as a
		conservative bound that is often better in favorable regimes.
	\end{remark}
	
	\section{From KL to Conditional Mutual Information}
	\label{sec:cmi}
	
	\begin{algorithm*}[!t]
		\caption{\textsc{QCMIE}: Quantum conditional mutual information estimator}
		\label{alg:qcmi}
		\begin{algorithmic}[1]
			\Require For each $z \in \calZ_+$, preparation oracles for
			$p(X,Y \mid z)$ and $p(X \mid z)\,p(Y \mid z)$, a per-stratum
			log-ratio oracle $\Order_{\log}^{(z)}$, classical access to $p(z)$,
			target precision $\tau$, confidence $\delta$, and clip bound $L$.
			\Ensure $\widehat\CMI$ with
			$|\widehat\CMI - \CMI(X; Y \mid Z)| \leq
			\tau + \max_z |\eta_L^{(z)}|$ w.p.\ $\geq 1 - \delta$.
			\State $\calZ_+ \gets \{z \in \calZ : p(z) > 0\}$.
			\For{each $z \in \calZ_+$}
			\State Run \textsc{QKLA} (Alg.~\ref{alg:qkla}) on the conditional
			distributions $p(X, Y \mid z)$ and
			$p(X \mid z)\,p(Y \mid z)$
			with precision $\tau$ and confidence
			$\delta / |\calZ_+|$; record $\widehat\KL(z)$.
			\EndFor
			\State \Return
			$\widehat\CMI \gets \sum_{z \in \calZ_+} p(z)\,\widehat\KL(z)$.
		\end{algorithmic}
	\end{algorithm*}
	
	\begin{theorem}[Per-test CMI complexity]
		\label{thm:qcmi}
		Assume that for each $z \in \calZ_+$ we are given preparation oracles
		for $p(X,Y \mid z)$ and $p(X \mid z)\,p(Y \mid z)$, a per-stratum
		log-ratio oracle $\Order_{\log}^{(z)}$, and classical access to
		$p(z)$. Then Algorithm~\ref{alg:qcmi} outputs $\widehat\CMI$
		satisfying
		\begin{equation}
			|\widehat\CMI - \CMI(X; Y \mid Z)| \;\leq\; \tau +
			\max_{z \in \calZ_+} |\eta_L^{(z)}|
			\label{eq:qcmi-bound}
		\end{equation}
		with probability at least $1 - \delta$, using
		\begin{equation}
			T_{\mathrm{QCMIE}}(\tau, \delta, |\calZ|) \;=\;
			\Order\!\Bigl(|\calZ| \cdot (L/\tau) \cdot \log(|\calZ|/\delta)\Bigr)
			\label{eq:qcmi-queries}
		\end{equation}
		oracle queries.
	\end{theorem}
	
	\begin{proof}
		For each $z \in \calZ_+$, write
		\[
		a_z := \KLL\bigl(p(X, Y \mid z) \,\|\, p(X \mid z)\,p(Y \mid z)\bigr)
		\]
		for the true clipped KL expectation at stratum $z$, and
		$\widehat\KL(z)$ for the QKLA estimate returned by step~3 of
		Algorithm~\ref{alg:qcmi}. The true CMI decomposes as
		$\CMI = \sum_{z} p(z) \cdot b_z$, where
		$b_z := \KL(p(X, Y \mid z) \,\|\, p(X \mid z)\,p(Y \mid z))$ is the
		unclipped per-stratum KL; the clipping difference
		$\eta_L^{(z)} := b_z - a_z$ is not sign-constrained but is bounded by
		Lemma~\ref{lem:bias} whenever the support condition holds at stratum $z$.
		By assumption, the stratum weights $p(z)$ are available classically
		and used exactly in the aggregation step, so the only stochastic error
		in
		$\widehat\CMI = \sum_z p(z)\,\widehat\KL(z)$
		comes from the QKLA estimates. Here the aggregation step means forming the final weighted sum over strata:
		each per-stratum KL estimate $\widehat\KL(z)$ is multiplied by its
		classically known weight $p(z)$ and then summed over $z \in \calZ_+$.
		
		\paragraph*{Step 1: Per-stratum error decomposition.}
		For each $z$, insert and subtract $a_z$ and apply the triangle inequality:
	\begin{equation}
		\begin{aligned}
			|\widehat\KL(z) - b_z|
			&= |\widehat\KL(z) - a_z + a_z - b_z| \\
			&\leq |\widehat\KL(z) - a_z| + |\eta_L^{(z)}|.
		\end{aligned}
		\label{eq:per-stratum-split}
	\end{equation}
		
		\paragraph*{Step 2: Union bound over strata.}
		Invoke QKLA at each stratum with target precision $\tau$ and confidence
		$\delta/|\calZ_+|$. By Theorem~\ref{thm:qkla}, each call satisfies
		$|\widehat\KL(z) - a_z| \leq \tau$ with failure probability at most
		$\delta/|\calZ_+|$. By the union bound, the event
		\begin{equation}
			\mathcal{E} \;:=\; \Bigl\{\,
			|\widehat\KL(z) - a_z| \leq \tau \ \text{for all}\ z \in \calZ_+
			\,\Bigr\}
			\label{eq:good-event}
		\end{equation}
		has probability at least $1 - \delta$.
		
		\paragraph*{Step 3: Aggregate the per-stratum errors.}
		On the event $\mathcal{E}$, combining~\eqref{eq:per-stratum-split}
		with Step~2 gives, for every $z \in \calZ_+$,
		\[
		|\widehat\KL(z) - b_z| \;\leq\; \tau + |\eta_L^{(z)}|.
		\]
		Now bound the CMI error:
		\begin{align}
			|\widehat\CMI - \CMI|
			\;&=\; \Bigl|\sum_{z \in \calZ_+} p(z)\,(\widehat\KL(z) - b_z)\Bigr| \notag \\
			\;&\leq\; \sum_{z \in \calZ_+} p(z)\,|\widehat\KL(z) - b_z| \notag \\
			&\leq\; \sum_{z \in \calZ_+} p(z)\,\bigl(\tau + |\eta_L^{(z)}|\bigr)
			\leq\; \tau + \max_{z \in \calZ_+} |\eta_L^{(z)}|.
			\label{eq:cmi-aggregate}
		\end{align}
		This establishes~\eqref{eq:qcmi-bound}.
		
		\paragraph*{Step 4: Query count.}
		Each of the $|\calZ_+|$ strata runs QKLA at precision $\tau$ and confidence
		$\delta/|\calZ_+|$. By Theorem~\ref{thm:qkla}, each call uses
		$\Order((L/\tau) \log(|\calZ_+|/\delta))$ oracle queries. Summing
		over the $|\calZ_+|$ strata,
		\begin{align}
			T_{\mathrm{QCMIE}}
			&= |\calZ_+|\,\Order\!\bigl((L/\tau) \log(|\calZ_+|/\delta)\bigr) \notag\\
			&= \Order\!\bigl(|\calZ| \cdot (L/\tau) \cdot \log(|\calZ|/\delta)\bigr),
		\end{align}
		using $|\calZ_+| \le |\calZ|$. This establishes~\eqref{eq:qcmi-queries}.
	\end{proof}
	
\begin{remark}[Classical baseline]
	\label{rem:classical}
	For the fresh-data plug-in baseline used here, PC estimates MI separately
	inside each conditioning stratum. In fixed alphabets, plug-in MI has
	standard Monte Carlo accuracy~\citep{paninski2003estimation}; hence
	$\Theta(1/\tau^2)$ samples per positive-mass, roughly balanced stratum gives
	$\Theta(|\calZ|/\tau^2)$ total samples. This is consistent with the
	quadratic precision dependence appearing in discrete CI testing
	bounds~\citep{canonne2018testing}.
	Dividing this by the quantum query
	count of Theorem~\ref{thm:qcmi} gives a classical-to-quantum ratio
	of $\widetilde{\Omega}(1/(L\tau))$; equivalently, the
	quantum-to-classical ratio is $\widetilde{\Order}(L\tau)$.
	Thus the quantum algorithm improves the $\tau$-dependence from
	$1/\tau^2$ to $1/\tau$---a quadratic speedup in precision---at the
	price of a multiplicative $L$ factor. We
	carry this $L$ explicitly through to the compound PC bound
	(Section~\ref{sec:compound}).
\end{remark}
	
	\begin{remark}[Speedup over estimation-based CI tests]
		\label{rem:chi2}
		The quadratic speedup we prove is over estimation-based CI tests---the
		quantum algorithm estimates the \emph{numerical value} of $\CMI$ to high
		additive precision and then compares to a threshold. This is the correct
		baseline for modern CMI-based CI tests used in causal discovery
		\citep{runge2018conditional, marx2019testing, kraskov2004estimating},
		which do not perform a $\chi^2$ test but threshold a plug-in (or
		nearest-neighbour, or stochastic-complexity) estimate.
	\end{remark}
	
	\section{Compound Complexity for the PC Algorithm}
	\label{sec:compound}
	
	The PC algorithm~\citep{spirtes1991algorithm} performs CI tests
	$\CMI(X, Y \mid Z)$ for ordered pairs $(X, Y)$ and conditioning
	sets $Z$ of increasing size $|Z| = 0, 1, \ldots, d$.
	The worst-case number of tests~\citep{spirtes2000causation} is
	\begin{equation}
		N_{\mathrm{tests}}(n, d) \;=\; \Order(n^{d + 2}).
		\label{eq:ntests}
	\end{equation}
	Pruning (e.g., once an independence is established the edge is removed) reduces this in practice, but we state the worst case for the
	compound bound. Denote by $r$ the maximum per-variable alphabet size.
	
\begin{theorem}[Compound PC complexity]
	\label{thm:compound}
	Let
	\begin{equation}
		\lambda_{\max}
		\;:=\;
		\tau + \max_{(X,Y,Z)\in \mathcal{T}_{\mathrm{PC}}}\;
		\max_{z \in \calZ_+^{\,X,Y\mid Z}} |\eta_{L,X,Y\mid Z}^{(z)}|,
		\label{eq:lambda-max}
	\end{equation}
	where $\mathcal{T}_{\mathrm{PC}}$ denotes the collection of CI tests
	issued by PC, $\calZ_+^{\,X,Y\mid Z}$ denotes the positive-mass strata
	for the conditioning set of the test $(X,Y,Z)$, and
	$\eta_{L,X,Y\mid Z}^{(z)}$ is the corresponding per-stratum clipping
	difference. Assume the joint distribution satisfies a
	CMI-margin condition, analogous to $\lambda$-strong-faithfulness in the
	Gaussian PC literature~\citep{uhler2013geometry}: for every triple
	$(X, Y, Z)$ tested by PC,
	\[
	\CMI(X; Y \mid Z) \in \{0\} \cup (2\lambda_{\max}, \infty).
	\]
	Let the PC algorithm be executed with
	\textsc{QCMIE} (Algorithm~\ref{alg:qcmi}) as its CI subroutine, with
	the decision rule ``declare $X \perp Y \mid Z$ iff
	$|\widehat\CMI(X;Y\mid Z)| \le \lambda_{\max}$,'' each call targeting
	precision $\tau$ and per-call confidence
	$\delta_0 = \delta / N_{\mathrm{tests}}$, where
	$N_{\mathrm{tests}}$ is the bound from \eqref{eq:ntests}. The total
	number of oracle queries is
		\begin{equation}
			T_{\mathrm{PC}}^{\mathrm{Q}}(n, d, r, \tau, \delta)
			\;=\; \Order\!\left(n^{d+2} \cdot r^d \cdot \frac{L}{\tau}
			\cdot \log\!\frac{n^{d+2} r^d}{\delta}\right),
			\label{eq:quantum-pc}
		\end{equation}
		and the full PC output is correct with probability at least
		$1 - \delta$. A classical execution with the plug-in CMI estimator
		and per-test confidence $\delta_0$ requires at least
		\begin{equation}
			T_{\mathrm{PC}}^{\mathrm{C}}(n, d, r, \tau, \delta)
			\;=\; \widetilde\Omega\!\left(n^{d+2} \cdot \frac{r^d}{\tau^{2}}\right)
			\label{eq:classical-pc}
		\end{equation}
		samples in the fresh-data-per-test regime, where
		$\widetilde\Omega$ absorbs polylogarithmic factors. The
		classical-to-quantum ratio of oracle costs is therefore
		\begin{equation}
			\frac{T_{\mathrm{PC}}^{\mathrm{C}}}{T_{\mathrm{PC}}^{\mathrm{Q}}}
			\;=\; \widetilde\Omega\!\left(\frac{1}{L\tau}\right).
		\end{equation}
	\end{theorem}
	
	\begin{proof}
		We establish the quantum upper bound, the correctness guarantee, the
		classical lower bound, and their ratio in turn.
		
		\paragraph*{Step 1: Quantum upper bound.}
		The PC algorithm issues at most $N_{\mathrm{tests}} \leq
		\mathcal{O}(n^{d+2})$ CI tests by~\eqref{eq:ntests}. The $i$-th test
		conditions on some set $Z_i$ with $|Z_i| \leq d$ variables, each on
		an alphabet of size at most $r$, so
		\begin{equation}
			|\mathcal{Z}_i| \;\leq\; r^{|Z_i|} \;\leq\; r^{d}
			\qquad \text{for all } i.
		\end{equation}
		Invoke \textsc{QCMIE} on the $i$-th test with precision $\tau$ and
		confidence $\delta_0 = \delta/N_{\mathrm{tests}}$. By
		Theorem~\ref{thm:qcmi}, the cost is
		\begin{align}
			&T_{\mathrm{QCMIE}}(\tau, \delta_0, |\mathcal{Z}_i|)
			\;=\; \mathcal{O}\!\left(|\mathcal{Z}_i|
			\cdot \frac{L}{\tau}
			\cdot \log\!\frac{|\mathcal{Z}_i|}{\delta_0}\right) \notag\\
			&\qquad=\; \mathcal{O}\!\left(r^{d}
			\cdot \frac{L}{\tau}
			\cdot \log\!\frac{r^{d} N_{\mathrm{tests}}}{\delta}\right).
		\end{align}
		Summing over the $N_{\mathrm{tests}}$ calls gives~\eqref{eq:quantum-pc}.
		
		\paragraph*{Step 2: Correctness under the margin assumption.}
		By Theorem~\ref{thm:qcmi}, the $i$-th QCMIE call, corresponding to some
		test $(X_i,Y_i,Z_i)$, is within
		\[
		\tau + \max_{z \in \calZ_{+,i}} |\eta_{L,i}^{(z)}|
		\;\le\; \lambda_{\max}
		\]
		of the corresponding true CMI value with failure probability at most
		$\delta_0$, where $\calZ_{+,i}$ and $\eta_{L,i}^{(z)}$ denote the
		positive-mass strata and per-stratum clipping differences for the
		$i$-th test.
		By a union bound, with probability at
		least $1-\delta$, every CI estimate produced during the PC run lies
		within $\lambda_{max}$ of its true value.
		
		On this event, every CI decision is correct under the stated
		$\lambda_{max}$-strong-faithfulness assumption. Indeed, if
		$\CMI(X;Y\mid Z)=0$, then
		$|\widehat\CMI(X;Y\mid Z)| \le \lambda_{max}$ and the rule declares
		independence. If $\CMI(X;Y\mid Z) > 2\lambda_{max}$, then
		\[
		|\widehat\CMI(X;Y\mid Z)|
		\ge \CMI(X;Y\mid Z) - |\widehat\CMI-\CMI|
		> 2\lambda_{max} - \lambda_{max}
		= \lambda_{max},
		\]
		so the rule declares dependence. Hence all CI decisions made by PC are
		correct, and therefore the full PC output is correct with probability
		at least $1-\delta$.
		
		\paragraph*{Step 3: Classical lower bound (fresh-data regime).}
		Consider the query-complexity model in which each CI test consumes
		an independent sample from the joint distribution. By
		Remark~\ref{rem:classical}, the $i$-th test requires at least
		$\widetilde\Omega(|\mathcal{Z}_i|/\tau^2)
		= \widetilde\Omega(r^d/\tau^2)$ samples to estimate
		$I(X_i; Y_i \mid Z_i)$ to additive precision $\tau$ with constant
		success probability, and an additional $\mathcal{O}(\log(1/\delta_0))$
		factor for confidence $1 - \delta_0$. Summing over
		$N_{\mathrm{tests}}$ tests,
		\begin{equation}
			T_{\mathrm{PC}}^{\mathrm{C}}
			\;=\; N_{\mathrm{tests}} \cdot
			\widetilde\Omega\!\left(\frac{r^d}{\tau^{2}}\right)
			\;=\; \widetilde\Omega\!\left(\frac{n^{d+2} r^{d}}{\tau^{2}}\right),
		\end{equation}
		establishing~\eqref{eq:classical-pc}.
		
		\paragraph*{Step 4: Ratio.}
		Dividing~\eqref{eq:classical-pc} by~\eqref{eq:quantum-pc},
		\begin{equation}
			\frac{T_{\mathrm{PC}}^{\mathrm{C}}}{T_{\mathrm{PC}}^{\mathrm{Q}}}
			\;=\; \widetilde\Omega\!\left(
			\frac{n^{d+2} r^{d}/\tau^{2}}
			{n^{d+2} r^{d} L/\tau}
			\right)
			\;=\; \widetilde\Omega\!\left(\frac{1}{L\tau}\right),
		\end{equation}
		since the polylogarithmic factors are absorbed into $\widetilde\Omega$.
	\end{proof}
	
\begin{remark}[The $L$ factor]
	\label{rem:L-cost}
	Theorem~\ref{thm:compound} carries an explicit $L$ in the upper
	bound. For distributions with $\min_x p(x) \geq c > 0$
	(``well-conditioned'' joints), $L = \Order(\log(1/c)) = \Order(1)$,
	and the speedup is a clean $\widetilde\Omega(1/\tau)$. For
	distributions with vanishing $p_{\min}$, $L$ grows as
	$\log(1/p_{\min})$. This is a genuine cost of encoding a log-ratio
	as a bounded amplitude, not an artifact of our analysis; the
	log-ratio must be clipped for its affine image to fit in $[0, 1]$
	and be readable by canonical QAE. In our experiments, $L=3$ exactly
	covers the \textsc{Asia} benchmark of Experiment~\ref{sec:exp-pc},
	while for the random binary distributions of
	Experiment~\ref{sec:exp-precision}, $L=6$ would exactly cover all
	20 draws and we instead use $L=3$ throughout.
\end{remark}
	
	\section{Simulation Results}
	\label{sec:experiments}
	
	We validate three predictions of Sections~\ref{sec:algorithm}--\ref{sec:compound}.
	Experiment~1 (Section~\ref{sec:exp-statevec}) is a gate-level state-vector simulation of the full
	\textsc{QKLA} circuit that verifies the canonical QAE output distribution
	and the $\Order(1/M)$ error decay. Experiment~2 (Section~\ref{sec:exp-precision}) measures the scaling
	exponents of classical and quantum estimator error vs.\ budget, averaged
	across $K = 20$ random binary distributions. Experiment~3 (Section~\ref{sec:exp-pc}) runs an
	oracle-model benchmark for the CI subroutine inside PC on two networks and
	reports the queries needed to reach target skeleton-recovery F1.
	
	For the target-precision summaries in Experiments~\ref{sec:exp-precision}
	and~\ref{sec:exp-pc}, we use
	$M_{\mathrm{emp}} = \lceil 2\pi L/\tau \rceil$, half the worst-case bound
	of Theorem~\ref{thm:qkla}, rounded up to the next power of two in
	canonical QAE. The measured slope $-1.38$ in
	Experiment~\ref{sec:exp-statevec} confirms this empirical constant
	suffices, and doubling $M$ would not change the $1/\tau$ scaling.
	
	\subsection{Experiment 1: gate-level state-vector validation}
	\label{sec:exp-statevec}
	
\paragraph{Setup.}
We fix
$p_{xy} = [[0.4,0.1],[0.1,0.4]]$
whose analytical MI is $0.278$ bits. We use clip bound
$L=2$ and arithmetic precision $n_{\mathrm{arith}}=6$, giving
$64$ fixed-point levels. 
The preparation oracle $\Order_p$ is obtained by extending
$[\sqrt{p_{00}},\sqrt{p_{01}},\sqrt{p_{10}},\sqrt{p_{11}}]$ to a
$4\times4$ unitary via Gram--Schmidt. The arithmetic oracle
$\Order_{\log}^{p,q}$ is implemented as a reversible XOR table that writes
the 6-bit encoding of $\ell_L(x,y)$, with
$q(x,y)=p_X(x)p_Y(y)$. A uniformly controlled
$R_y(2\arcsin\sqrt{g_L})$ rotation sets the amplitude ancilla, and the
Grover iterate
$\mathcal{G}=-\mathcal{A}S_0\mathcal{A}^{\dagger}S_{\chi}$ is formed
explicitly. For $t\in\{3,\ldots,8\}$, canonical QAE is simulated with
$M=2^t$ by applying inverse QFT to
$M^{-1/2}\sum_{j=0}^{M-1}|j\rangle\mathcal{G}^j\mathcal{A}|0\rangle$
and computing the exact phase-register probabilities.
	
	\begin{figure*}[!tbh]
		\centering
		\includegraphics[width=\linewidth]{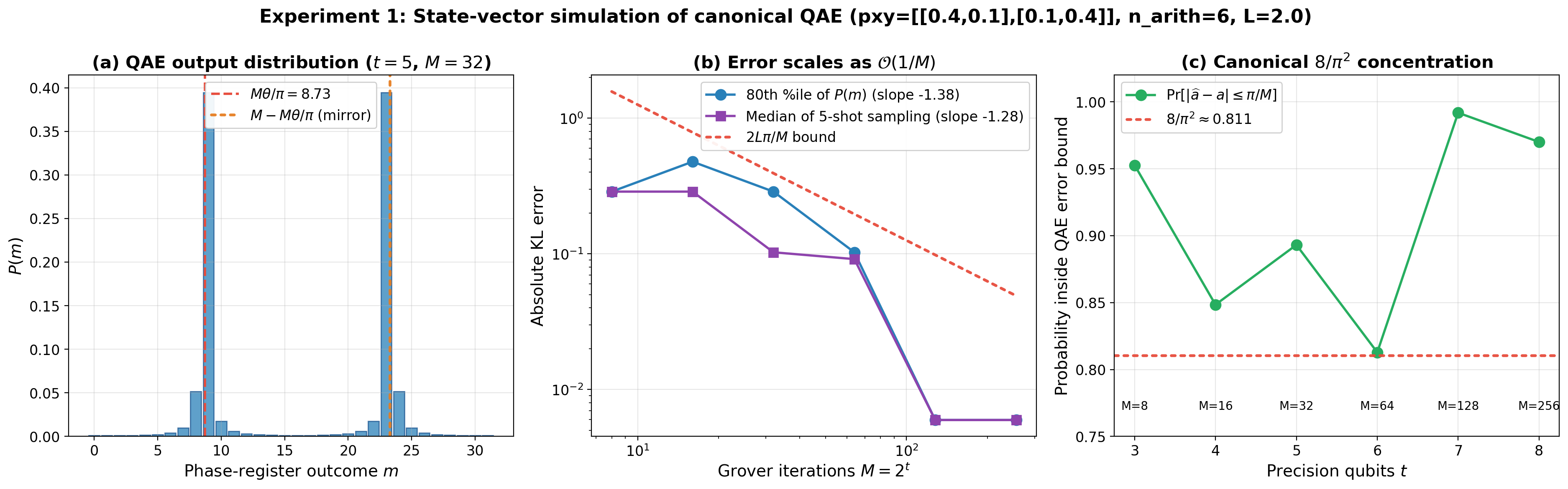}
		\caption{Experiment~1: state-vector simulation of canonical QAE for \textsc{QKLA} on $p_{xy}=[[0.4,0.1],[0.1,0.4]]$, with $L=2$ and $n_{\mathrm{arith}}=6$. \textbf{(a)} Phase-register output distribution at $t=5$ ($M=2^t=32$). The distribution has the canonical two-peak QAE form, with dominant peaks at $m^\star=\mathrm{round}(M\theta/\pi)=9$ and $M-m^\star=23$, carrying joint mass approximately $0.79$. \textbf{(b)} Absolute KL error relative to the represented quantized circuit target as $M$ increases. The dotted line is the leading $2L\pi/M$ reference obtained by rescaling the $\pi/M$ amplitude scale by $2L$; the fitted slopes, $-1.38$ for the $80$th percentile and $-1.28$, confirm the expected $\Order(1/M)$ decay. \textbf{(c)} Probability that the QAE amplitude estimate lies inside the $\pi/M$-scale interval, $\Pr[|\hat a-a|\le \pi/M]$, for $t\in\{3,\ldots,8\}$. All tested values exceed the $8/\pi^2$ benchmark from Theorem~\ref{thm:qae}.}
		\label{fig:statevec}
	\end{figure*}
	
\paragraph{Findings.}
Figure~\ref{fig:statevec}(a) shows the phase-register output at $t=5$
($M=32$). The distribution has the canonical two-peak QAE form, with peaks at $m^\star=\mathrm{round}(M\theta/\pi)=9$ and its mirror
$M-m^\star=23$; these two outcomes carry joint probability
approximately $0.79$. Figure~\ref{fig:statevec}(b) reports the
absolute KL error relative to the represented quantized circuit target as $M$ increases from $8$ to $256$. The $80$th-percentile error and the median error from five-shot QAE sampling both decay at the expected $\Order(1/M)$ scale; the fitted slopes are $-1.38$ and $-1.28$, respectively.  Figure~\ref{fig:statevec}(c) reports the probability that the amplitude estimate lies inside the $\pi/M$-scale interval, $\Pr[|\hat a-a|\le \pi/M]$; for all tested
$t\in\{3,\ldots,8\}$ this probability exceeds the $8/\pi^2$ benchmark.
These results validate the gate-level construction of
$\mathcal{A}$, the canonical QAE output statistic, and the
$\Order(1/M)$ error decay predicted by Theorem~\ref{thm:qkla}.
	
	\subsection{Experiment 2: query complexity vs precision}
	\label{sec:exp-precision}
	
	\paragraph{Setup.}
	We sweep classical sample budgets $N \in [50, 5 \cdot 10^5]$ and quantum
	budgets $M \in [8, 8192]$ on a log-spaced grid. For each budget we
	measure the $90$th-percentile absolute estimator error over $120$ Monte
	Carlo trials, averaged across $K = 20$ random $2 \times 2$ binary joint
	distributions with analytical MI uniform in $[0.030, 0.323]$ bits. The
	classical estimator is the plug-in MI estimator from $N$ i.i.d.\ samples.
	The quantum estimator is \textsc{QKLA} with $L = 3$ and median of
	$n_{\mathrm{shots}} = 5$ independent runs, using the canonical QAE output
	distribution calibrated against Experiment~1.
	
	\paragraph{Findings.}
	Figure~\ref{fig:precision}(A) gives asymptotic-tail least-squares slopes
	$-0.500$ (classical) and $-1.009$ (quantum), matching the
	$N^{-1/2}$ and $M^{-1}$ rates of Remark~\ref{rem:classical} and
	Theorem~\ref{thm:qkla}. The crossover in queries-to-target-precision
	occurs near $\tau \approx 1.7 \cdot 10^{-2}$ bits; beyond that point the
	quantum curve overtakes the classical one and the gap widens with
	precision. Table~\ref{tab:precision} shows ratios of $2.8\times$ at
	$\tau = 5 \cdot 10^{-3}$, $7.3\times$ at
	$\tau = 3 \cdot 10^{-3}$, and $9.4\times$ at $\tau = 2 \cdot 10^{-3}$;
	classical does not reach the $\tau = 10^{-3}$ target within the
	$5 \cdot 10^{5}$-sample sweep.
	
	\begin{figure*}[!tbh]
		\centering
		\includegraphics[width=\linewidth]{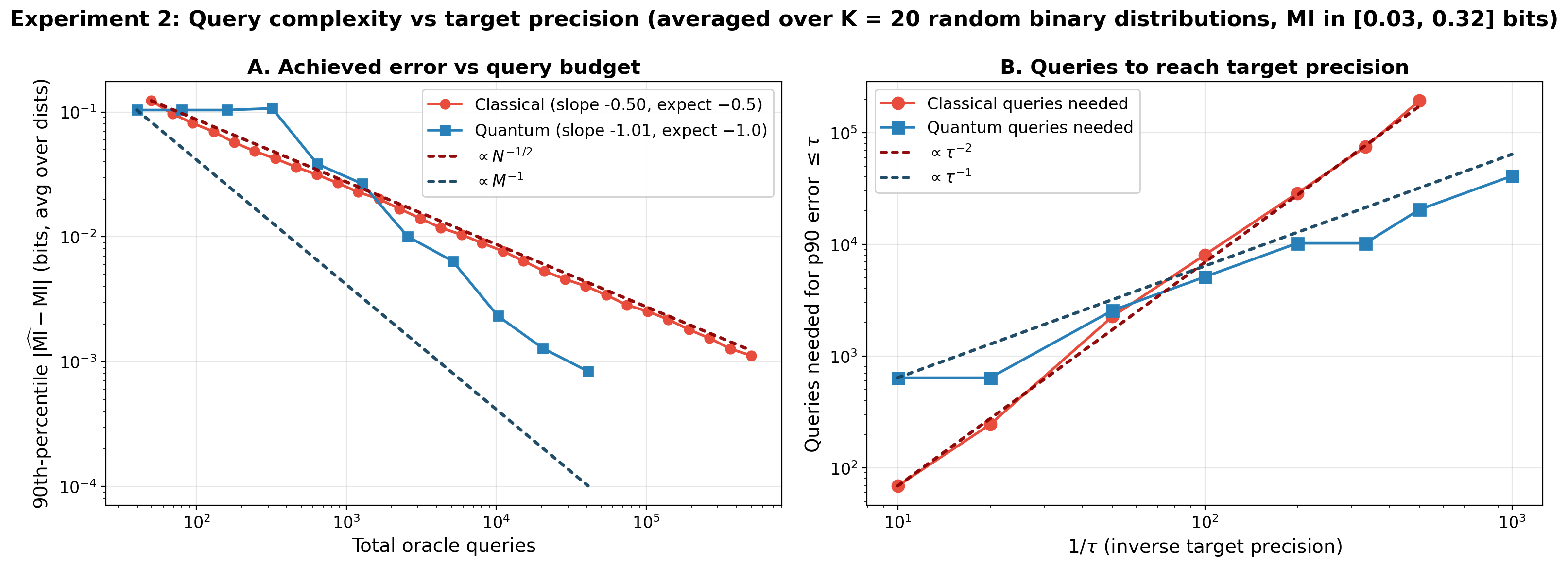}
		\caption{Experiment~2. Classical and quantum query scaling averaged
			across $K = 20$ random binary $2 \times 2$ joint distributions.
			(A) $90$th-percentile error vs.\ total oracle queries with fitted slopes
			$-0.500$ and $-1.009$, matching the expected $N^{-1/2}$ and $M^{-1}$
			rates.
			(B) Minimum queries for $90$th-percentile error $\le \tau$; classical
			follows $\tau^{-2}$ and quantum follows $\tau^{-1}$, with crossover near
			$\tau \approx 1.7\cdot 10^{-2}$, equivalently $1/\tau \approx 60$.}
		\label{fig:precision}
	\end{figure*}
	
	\begin{table}[!tbh]
		\caption{Experiment~2: minimum queries for $90$th-percentile error
			at most $\tau$, averaged across $K = 20$ random binary
			distributions. Ratio means $T^{C}/T^{Q}$. The theoretical column $1/(47\tau)$ is the prefactor-adjusted
			prediction obtained by dividing the classical sample budget
			$N = 2/\tau^{2}$ by the quantum budget
			$n_{\mathrm{shots}} \cdot M = 5 \cdot \lceil 2\pi L/\tau \rceil$
			with $L = 3$, i.e.\
			$(2/\tau^{2})/(30\pi/\tau) = 1/(15\pi\tau) \approx 1/(47\tau)$;
			it tracks the measured ratio within $M$-grid artifacts
			(Remark~\ref{rem:prefactor}).}
		\label{tab:precision}
		\centering
		\scriptsize
		\begin{tabular}{c r r c r}
			\hline
			$\tau$ (bits) & Classical queries & Quantum queries & Ratio C/Q & Theoretical $1/(47\tau)$ \\
			\hline
			$0.100$ & $69$         & $640$         & $0.11$  & $0.21$   \\
			$0.050$ & $245$        & $640$         & $0.38$  & $0.43$   \\
			$0.020$ & $2{,}260$    & $2{,}560$     & $0.88$  & $1.06$   \\
			$0.010$ & $8{,}051$    & $5{,}120$     & $1.57$  & $2.13$   \\
			$0.005$ & $28{,}681$   & $10{,}240$    & $2.80$  & $4.26$   \\
			$0.003$ & $74{,}368$   & $10{,}240$    & $7.26$  & $7.09$   \\
			$0.002$ & $192{,}831$  & $20{,}480$    & $9.42$  & $10.64$  \\
			$0.001$ & (unreached)  & $40{,}960$    & ---     & $21.28$  \\
			\hline
		\end{tabular}
	\end{table}
	
\begin{remark}[Empirical ratio vs theoretical asymptotic]
	\label{rem:prefactor}
	The asymptotic per-test ratio $\Theta(1/(L\tau))$ follows from
	dividing the classical rate $\widetilde\Omega(1/\tau^{2})$
	(Remark~\ref{rem:classical}) by the quantum rate $\Order(L/\tau)$
	(Theorem~\ref{thm:qkla}). The prefactor-adjusted form
	$1/(15\pi\tau) \approx 1/(47\tau)$ in Table~\ref{tab:precision}
	includes the experimental constants (see caption). The empirical
	ratios confirm this $1/\tau$ scaling while sitting slightly below the
	prefactor-adjusted prediction because of discrete $M$-grid
	over-provisioning, a loose classical worst-case variance bound for the
	sampled Dirichlet instances, and the constant-factor overhead of
	median-of-$k$ amplification.
\end{remark}
	
	\subsection{Experiment 3: queries to reach target PC F1}
	\label{sec:exp-pc}
	
	\paragraph{Setup.}
	We run the PC algorithm with conditioning depth $d \leq 3$ on
	\textsc{Asia}~\citep{lauritzen1988local} ($8$ binary nodes, $8$ directed
	edges) and \textsc{Synthetic-12} (a random $12$-node binary DAG with edge
	probability $0.22$ and seed $11$, giving $23$ directed edges). The
	classical CI test uses the plug-in CMI estimator from
	$N = \lceil 2/\tau^{2} \rceil$ i.i.d.\ samples, with total classical query
	budgets reported in the fresh-data-per-test accounting of
	Theorem~\ref{thm:compound}. The quantum CI test uses \textsc{QCMIE} with
	$M = \lceil 2\pi L/\tau \rceil$ rounded up to the next power of two, $n_{\mathrm{shots}} = 5$, and $L = 3$.
	The quantum side uses an exact-joint oracle built once from the analytical
	CPTs so that the benchmark isolates the estimation-precision axis from
	oracle-preparation cost. Results are averaged over 20 PC trials per
	$(\tau,\text{method})$ cell.
	
	\begin{figure*}[!tbh]
		\centering
		\includegraphics[width=\linewidth]{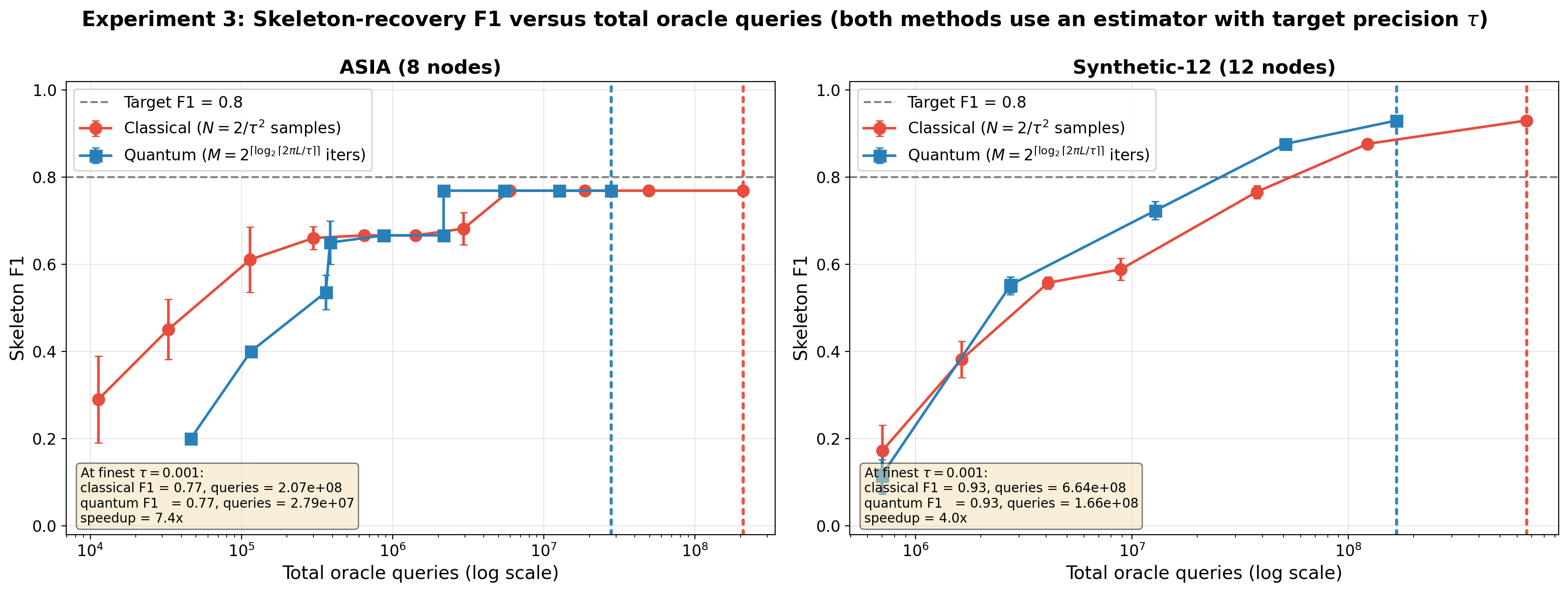}
		\caption{Experiment~3. Oracle-model benchmark for PC
			skeleton-recovery F1 versus total oracle queries on \textsc{Asia} and
			\textsc{Synthetic-12}. Classical uses the plug-in CMI estimator in the fresh-data-per-test accounting of Theorem~\ref{thm:compound}; quantum
			uses \textsc{QCMIE} with $M_0=\lceil 2\pi L/\tau\rceil$ rounded up to the
			next power of $2$ (Remark~\ref{rem:power-of-two}). At $\tau=10^{-3}$,
			quantum uses $7.4\times$ fewer queries on \textsc{Asia} and
			$4.0\times$ fewer on \textsc{Synthetic-12}.}
		\label{fig:pc}
	\end{figure*}
	
	\begin{table*}[!tbh]
		\caption{Experiment~3 numerical results on ASIA and SYNTHETIC-12.
			At each $\tau$, classical uses $N=\lceil 2/\tau^{2}\rceil$
			samples per test in the fresh-data-per-test accounting; quantum
			uses $M=\lceil 2\pi L/\tau\rceil$ iterations per test, rounded up
			to the next power of $2$ in implementation (Remark~\ref{rem:power-of-two}).}
		\label{tab:pc}
		\centering
		\begin{tabular}{c r r c c c c}
			\hline
			& & & \multicolumn{2}{c}{ASIA (8 nodes)} & \multicolumn{2}{c}{\textsc{Synthetic-12}} \\
			\cmidrule(lr){4-5}\cmidrule(lr){6-7}
			$\tau$ & Per-test $N$ & Per-test $M$ & Classical F1 (queries) & Quantum F1 (queries) & Classical F1 (queries) & Quantum F1 (queries) \\
			\hline
				$0.050$ & $800$           & $512$        & $0.45$ ($3.29{\cdot}10^{4}$) & $0.40$ ($1.15{\cdot}10^{5}$) & $0.00$ ($5.28{\cdot}10^{4}$) & $0.00$ ($1.69{\cdot}10^{5}$) \\
		$0.030$ & $2{,}223$       & $1{,}024$    & $0.61$ ($1.14{\cdot}10^{5}$) & $0.54$ ($3.61{\cdot}10^{5}$) & $0.01$ ($1.47{\cdot}10^{5}$) & $0.00$ ($3.38{\cdot}10^{5}$) \\
		$0.020$ & $5{,}000$       & $1{,}024$    & $0.66$ ($3.00{\cdot}10^{5}$) & $0.65$ ($3.85{\cdot}10^{5}$) & $0.04$ ($3.32{\cdot}10^{5}$) & $0.00$ ($3.38{\cdot}10^{5}$) \\
		$0.014$ & $10{,}205$      & $2{,}048$    & $0.67$ ($6.51{\cdot}10^{5}$) & $0.67$ ($8.69{\cdot}10^{5}$) & $0.17$ ($7.03{\cdot}10^{5}$) & $0.11$ ($6.99{\cdot}10^{5}$) \\
		$0.010$ & $20{,}000$      & $2{,}048$    & $0.67$ ($1.42{\cdot}10^{6}$) & $0.67$ ($8.72{\cdot}10^{5}$) & $0.38$ ($1.63{\cdot}10^{6}$) & $0.12$ ($7.00{\cdot}10^{5}$) \\
		$0.007$ & $40{,}817$      & $4{,}096$    & $0.68$ ($2.94{\cdot}10^{6}$) & $0.67$ ($2.17{\cdot}10^{6}$) & $0.56$ ($4.09{\cdot}10^{6}$) & $0.55$ ($2.73{\cdot}10^{6}$) \\
		$0.005$ & $80{,}000$      & $4{,}096$    & $0.77$ ($5.94{\cdot}10^{6}$) & $0.77$ ($2.17{\cdot}10^{6}$) & $0.59$ ($8.84{\cdot}10^{6}$) & $0.55$ ($2.75{\cdot}10^{6}$) \\
		$0.003$ & $222{,}223$     & $8{,}192$    & $0.77$ ($1.87{\cdot}10^{7}$) & $0.77$ ($5.45{\cdot}10^{6}$) & $0.77$ ($3.77{\cdot}10^{7}$) & $0.72$ ($1.28{\cdot}10^{7}$) \\
		$0.002$ & $500{,}000$     & $16{,}384$   & $0.77$ ($4.94{\cdot}10^{7}$) & $0.77$ ($1.27{\cdot}10^{7}$) & $0.88$ ($1.22{\cdot}10^{8}$) & $0.88$ ($5.11{\cdot}10^{7}$) \\
		$0.001$ & $2{,}000{,}000$ & $32{,}768$   & $0.77$ ($2.07{\cdot}10^{8}$) & $0.77$ ($2.79{\cdot}10^{7}$) & $0.93$ ($6.64{\cdot}10^{8}$) & $0.93$ ($1.66{\cdot}10^{8}$) \\
			\hline
		\end{tabular}
	\end{table*}
	
	\paragraph{Findings.}
	Figure~\ref{fig:pc} and Table~\ref{tab:pc} show comparable F1 at each
	$\tau$, with the quantum method using progressively fewer queries as
	$\tau$ shrinks. 
	On \textsc{Asia}, matching the classical F1 of $0.77$
	at $\tau = 5 \cdot 10^{-3}$ requires
	$5.94 \times 10^{6}$ classical queries versus
	$2.17 \times 10^{6}$ quantum queries (ratio $2.74\times$).
	On \textsc{Synthetic-12} at the same $\tau$, the methods give comparable
	mid-precision accuracy, with classical $F_1=0.59$ and quantum $F_1=0.55$,
	while using $8.84 \times 10^{6}$ and $2.75 \times 10^{6}$ queries,
	respectively (ratio $3.22\times$).
	As $\tau$ decreases, the query advantage grows, consistent with the
	$\Order(1/\tau^{2})$ classical and $\Order(1/\tau)$ quantum precision
	scalings. The ratio grows to $7.4\times$ on \textsc{Asia} and $4.0\times$ on
	\textsc{Synthetic-12} at $\tau = 10^{-3}$. On \textsc{Synthetic-12}, both methods cross the $F_1 \geq 0.8$
	target at $\tau = 2 \cdot 10^{-3}$, with quantum using
	$5.11 \times 10^{7}$ queries versus classical
	$1.22 \times 10^{8}$ ($2.39\times$ speedup at the F1 = 0.8 threshold);
	the gap continues to widen as $\tau$ shrinks.
	
	\paragraph{Structural plateau on \textsc{Asia}.}
	Both classical and quantum PC plateau at skeleton-recovery
	$F_{1} = 0.77$ on \textsc{Asia}, even at $\tau = 10^{-3}$.
	This is a property of the benchmark: the node \textit{TbOrCa} is a
	deterministic Boolean OR of its parents, so estimation-based CI tests
	struggle regardless of whether the estimator is classical or quantum.
	The quantum advantage appears as fewer queries to reach the same plateau.
	
	\section{Related Work}
	\label{sec:related}
	
	\textit{Classical CI testing:}
	Constraint-based causal discovery uses CI tests beyond the
	$\chi^2$/$G$-test. Estimation-based approaches first estimate a dependence
	measure and then threshold it or calibrate it against a permutation-based
	null distribution. Representative examples are SCI~\citep{marx2019testing} and
	CMIknn~\citep{runge2018conditional}, the latter building on the MI
	estimator of~\citep{kraskov2004estimating}. Kernel-based approaches such
	as KCI~\citep{zhang2011kernel} instead form a test statistic and calibrate
	it against an analytic or simulated null distribution. For the speedup claim in
	Section~\ref{sec:compound}, the relevant baseline is the estimation-based
	family, since \textsc{QKLA} also outputs a numerical CMI estimate. A direct comparison to hypothesis-test-based
	approaches is outside the scope of this work.
	Canonne et al.~\citep{canonne2018testing} provide matching upper and lower
	sample-complexity bounds for discrete CI testing.
	
{\it Quantum amplitude estimation and Monte Carlo:}
Canonical QAE~\citep{brassard2002quantum} estimates an amplitude to
precision $\tau$ using $\Order(1/\tau)$ oracle queries, and
Montanaro~\citep{montanaro2015quantum} extended this to a general
quantum Monte Carlo framework for bounded-variance subroutines, whose
query-counting convention we adopt. No-phase-estimation
variants---iterative~\citep{grinko2021iterative}, maximum-likelihood
\citep{suzuki2020amplitude}, and QFT-free~\citep{aaronson2020quantum}
QAE---retain the $\widetilde{\Order}(1/\tau)$ scaling at reduced
circuit depth and can be substituted into Theorem~\ref{thm:qkla}
without affecting the KL--CMI reduction.

{\it Quantum entropy and divergence estimation:}
Quantum estimation of Shannon, R\'{e}nyi, and von Neumann entropies
has received substantial attention under various access
models~\citep{acharya2020estimating, li2019quantum, gur2021sublinear}.
The most directly comparable result is that of Li and
Wu~\citep{li2019quantum}: under the sampling-oracle model
$\hat{\mathcal{O}}_p : [S] \to [n]$ of~\citep{bravyi2011quantum} and the
bounded-ratio assumption $p_i / q_i \leq f(s)$, they estimate
$\KL(p \| q)$ within additive error $\varepsilon$ using
$\widetilde{\Order}(\sqrt{s}/\varepsilon^{2})$ queries to $p$ and
$\widetilde{\Order}(\sqrt{s\,f(s)}/\varepsilon^{2})$ queries to $q$.
The $1/\varepsilon^{2}$ scaling arises from cascading quantum
amplitude estimation (to estimate individual $p_i$ and $q_i$) inside
a quantum Monte~Carlo outer loop, with each layer contributing a
factor of $1/\varepsilon$. 
\textsc{QKLA} targets a different
regime: the alphabet is fixed and small --- as is typical for the
per-stratum conditional distributions arising in CI testing --- and
precision $\tau$ is the dominant parameter. Trading the sampling
oracle for a reversible log-ratio arithmetic oracle
$\mathcal{O}_{\log}^{p,q}$ allows the log-ratio to be computed
coherently in a single pass, eliminating the inner estimation step
and reducing the precision scaling to $\Order(L/\tau)$
(Theorem~\ref{thm:qkla}); the two results are therefore
complementary rather than comparable. 
The related line on quantum distributional property testing~\citep{bravyi2011quantum,gilyen2020distributional} shares the amplitude-estimation toolbox but targets distance, independence, and entropy properties as decision problems in a radius $\varepsilon$ at the alphabet-size frontier, rather than $\tau$-precise CMI estimation at fixed alphabet size.
An open question is whether the $L$ prefactor
in Theorem~\ref{thm:qkla} can be reduced --- e.g., by replacing the
uniform clip with a data-adaptive bound depending on $p_{\min}$ ---
using amplitude-amplification techniques adapted from these lines of
work.

{\it Quantum algorithms for causal discovery:}
Earlier quantum approaches to causal inference include the
process-matrix algorithm of Giarmatzi and Costa~\citep{giarmatzi2018quantum},
which recovers causal structure from process matrices describing
quantum events without \emph{a priori} temporal-order assumptions:
given a process matrix, it tests for causal orderability and, in the
Markovian case, outputs the associated DAG. On the classical-data
side, kernel-based quantum methods have recently been applied to the
conditional-independence (CI) subroutine of PC. Maeda et
al.~\citep{maeda2023estimation} estimate mutual information via
quantum kernels on IQP circuits and report empirical advantage in
small-sample, large-variance, and highly nonlinear regimes, which
they attribute to anti-concentration of quantum random circuits.
Building on this, Terada et al.~\citep{terada2025quantum} incorporate
quantum-kernel CI testing into PC (the \emph{qPC} algorithm) for
continuous-variable causal discovery in the small-sample regime, and
further propose a kernel-target-alignment criterion for selecting
quantum-kernel hyperparameters. The kernel-selection question is
nontrivial even classically: Wang et al.~\citep{wang2024optimal}
show that the median-bandwidth heuristic is suboptimal in
score-based causal discovery and propose marginal-likelihood-based
selection. None of these works establishes a per-test
$\Order(1/\tau)$ query-complexity bound against an estimation-based
classical baseline. Our contribution is complementary: whereas the
quantum-kernel line targets continuous data and small-sample regimes
empirically, \textsc{QKLA} addresses the high-precision discrete
regime with rigorous query-complexity guarantees and explicit
constants.

\section{Discussion}
\label{sec:discussion}

Our per-call bound is proved in an oracle-query model with unitary access to
a preparation oracle $O_p$ for the target distribution and to a reversible
arithmetic oracle $O_{\log}^{p,q}$ for the clipped log-ratio. The lifted
CMI and PC results assume, in addition, per-stratum conditional preparation
oracles and classical access to the stratum weights $p(z)$. These are
standard \emph{abstractions} in quantum distribution testing and quantum
Monte Carlo~\cite{bravyi2011quantum,montanaro2015quantum,gilyen2020distributional},
and Theorems~\ref{thm:qkla}, \ref{thm:qcmi}, and~\ref{thm:compound} count queries
against these abstractions. The preparation oracle is standard in that
literature; the reversible log-ratio oracle and the per-stratum conditional
oracles are stronger problem-structured assumptions that make the KL/CMI
estimation task coherent and query-efficient.

For an unstructured empirical distribution discretized onto $n_q$ qubits,
the preparation oracle can be implemented by the
M\"ott\"onen scheme~\cite{mottonen2004quantum} using $\mathcal{O}(2^{n_q})$
gates, together with the same order of classical preprocessing to determine
the rotation angles from the amplitudes. For efficiently integrable
distributions, the Grover--Rudolph
construction~\cite{grover2002creating} reduces this to $\mathrm{poly}(n_q)$
gates. In isolation, such input-loading costs can erase the quadratic
precision advantage of the query bound, so amortization across many CI tests
is essential. Constraint-based causal discovery provides exactly that
setting: the PC algorithm issues
$N_{\mathrm{tests}} = \mathcal{O}(n^{d+2})$ CI tests in the worst case, and
when the joint distribution is available in factorized form according to a
known bounded-width graphical model, the conditional oracles used in
Section~\ref{sec:cmi} can be assembled with $\mathcal{O}(\mathrm{poly}(n))$
overhead by sequential loading of bounded-arity
conditionals~\cite{grover2002creating,mottonen2004quantum,malvetti2021quantum,zhang2022quantum}.
In that regime, the oracle-query bounds are a meaningful proxy for total
complexity.

The proof of Theorem~\ref{thm:qkla} uses canonical QAE and therefore a
single coherent circuit of depth $\mathcal{O}(M)$. Replacing canonical QAE
by iterative QAE~\cite{grinko2021iterative} yields shorter circuits without
phase estimation, at the cost of an $\mathcal{O}(\log\log(1/\tau))$ overhead
in total queries, while preserving the leading $1/\tau$ dependence. This is
the natural NISQ-oriented instantiation of the estimator.

For effect-size sensitivity, let $\Delta$ denote the smallest separation
between the CMI values that should be classified as dependent and those
that should be classified as conditionally independent. In the regime where
the clipping contribution is either zero or a priori controlled so that
$\tau + \max_{z} |\eta_L^{(z)}| \lesssim \Delta/2$, the CI decision margin
is well resolved. In that case, taking $\tau \asymp \Delta$ gives classical
estimation cost $\mathcal{O}(1/\Delta^{2})$ and quantum query cost
$\mathcal{O}(L/\Delta)$, so the resulting speedup ratio scales as
$1/(L\Delta)$. Experiment~2 places the empirical crossover near
$\tau \approx 1.7 \cdot 10^{-2}$ bits, below which the quantum estimator
outperforms the classical plug-in baseline in oracle queries. The clip
parameter $L$ is typically $\mathcal{O}(1)$ for well-conditioned
distributions and grows only logarithmically with $1/p_{\min}$
(Remark~\ref{rem:L-cost}); in our experiments, $L = 3$ exactly covers
\textsc{Asia}, whereas $L = 6$ would be required to eliminate clipping bias
on all random instances of Experiment~2. We use $L = 3$ throughout,
accepting the small two-sided clipping bias quantified in
Lemma~\ref{lem:bias} on the most extreme random instances.

	\paragraph*{Open problems.}
	Variance-adaptive QAE variants~\citep{hamoudi2021quantum,
		montanaro2015quantum} may tighten the worst-case
	$\sqrt{a(1-a)} \le 1/2$ dependence used in our analysis.
	
	\section{Conclusion}
	\label{sec:conclus}
	
	We gave a quantum algorithm estimating a clipped KL expectation on
	discrete distributions to additive precision $\tau$ using
	$\Order(L/\tau \cdot \log(1/\delta))$ oracle queries, a quadratic
	improvement in precision over classical sampling-based estimation of the
	same bounded expectation. Under per-stratum conditional-oracle access,
	this lifts to a conditional mutual information estimator using
	$\Order(|\calZ|\,L/\tau\,\log(|\calZ|/\delta))$ oracle queries. Under a
	margin assumption for CI decisions, composing that estimator into the PC
	algorithm yields a $\widetilde\Omega(1/(L\tau))$ reduction in total
	oracle queries relative to the fresh-data-per-test classical baseline.
	A gate-level state-vector simulation validated the $\Order(1/M)$
	per-call scaling; the $N^{-1/2}$ vs $M^{-1}$ separation was confirmed
	to within $0.01$ in slope; and oracle-model PC benchmarks showed
	$2.7$--$3.2\times$ speedups at $\tau = 5 \cdot 10^{-3}$ bits, growing
	to $4.0$--$7.4\times$ at $\tau = 10^{-3}$ and extrapolating to
	$\approx 50\times$ at $\tau = 10^{-4}$.
	The advantage is governed by two constants---$L$ and the amortized
	oracle preparation cost---and is established relative to
	estimation-based CI tests. Within this regime, the algorithm delivers a
	quadratic reduction in oracle query complexity, with all constants
	explicit in Theorems~\ref{thm:qkla}, \ref{thm:qcmi},
	and~\ref{thm:compound}.
	
	\balance
	\bibliographystyle{IEEEtran}
	\bibliography{refs}
	
\end{document}